\documentclass[final]{elsarticle}
\usepackage{mathtools}
\usepackage{amsthm}
\usepackage{amsmath}
\usepackage{amssymb}
\usepackage[boxed,ruled,noend,noline,linesnumbered]{algorithm2e}
\usepackage{hyperref}
\usepackage{fixltx2e}
\usepackage{caption}
\usepackage[table]{xcolor}
\usepackage{chngcntr}
\usepackage{comment}
\biboptions{sort&compress}

\mathchardef\mhyphen="2D

\newcommand\placeqed{\tiny{\nobreak\enspace$\square$}}

\newtheorem{proposition}{Proposition}

\newenvironment{sproof}{\noindent \textit{Proof}.}

\makeatletter
\def\ps@pprintTitle{%
	\let\@oddhead\@empty
	\let\@evenhead\@empty
	\def\@oddfoot{}%
	\let\@evenfoot\@oddfoot}
\makeatother

\begin{document}

\begin{frontmatter}

\title{A Smart Backtracking Algorithm for Computing Set Partitions with Parts of Certain Sizes}

\author{Samer Nofal \\ \href{mailto:samer.nofal@gju.edu.jo}{samer.nofal@gju.edu.jo} }
\address{Department of Computer Science, German Jordanian University\\ 
P.O. Box 35247, Amman 11180, Jordan \\ Tel. +962 6 429 4444}

\begin{abstract}
Let $\alpha=\{a_1,a_2,a_3,...,a_n\}$ be a set of elements, $\delta < n$ be a non-negative integer, and $\Gamma: \alpha \to \{0, 1, 2, ..., n\}$ be a total mapping. Then, we call $\Gamma$ a \emph{partition} of $\alpha$ if and only if for all $x \in \alpha$, $\Gamma(x) \neq 0$. Further, we call $\Gamma$ a $\delta$-\emph{partition} of $\alpha$ if and only if $\Gamma$ is a partition of $\alpha$ and for all $i \in \{1, 2, 3, ..., n\}$, $|\{x: \Gamma(x)=i\}| > \delta$. We give a non-trivial algorithm that computes all $\delta$-partitions of $\alpha$ in $\Omega(n)$ time. On the opposite, a naive generate-and-test algorithm would compute all $\delta$-partitions of $\alpha$ in $\Omega(nB_n)$ time where $B_n$ is the Bell number.
\end{abstract}
\end{frontmatter}

\section{Introduction}
Let $\alpha=\{a_1,a_2,a_3,...,a_n\}$ be a set of elements, $\delta < n$ be a non-negative integer, and $\Gamma: \alpha \to \{0, 1, 2, ..., n\}$ be a total mapping. Then, we call $\Gamma$ a \emph{partition} of $\alpha$ if and only if for all $x \in \alpha$, $\Gamma(x) \neq 0$. Further, we call $\Gamma$ a $\delta$-\emph{partition} of $\alpha$ if and only if $\Gamma$ is a partition of $\alpha$ and for all $i \in \{1, 2, 3, ..., n\}$, $|\{x: \Gamma(x)=i\}| > \delta$. For example, the partitions of $\{a_1,a_2,a_3\}$ are $\{\{a_1,a_2,a_3\}\}$, $\{\{a_1,a_2\},\{a_3\}\}$, $\{\{a_1,a_3\},\{a_2\}\}$, $\{\{a_1\},\{a_2,a_3\}\}$, $\{\{a_1\},\{a_2\},\{a_3\}\}$. In contrast, the $1$-partition of $\{a_1,a_2,a_3\}$ is $\{\{a_1,a_2,a_3\}\}$, which is equal to the $2$-partition of $\{a_1,a_2,a_3\}$.~Note that every part of a partition has all of its elements mapped to the same integer.

Perhaps, the earliest work on formal algorithms for computing set partitions goes back to the article of \cite{DBLP:journals/cacm/Hutchinson63} that is published in the year 1963. After then, until the year 2013, one can see several proposals for different algorithms for computing set partitions, see the manuscripts of  \cite{DBLP:journals/jacm/Ehrlich73,DBLP:journals/ipl/Kaye76,DBLP:conf/isaac/Ruskey93,DBLP:journals/cj/Er88a,DBLP:journals/cj/DjokicMSSS89,DBLP:journals/jmma/MansourN08,Kokosinski,DBLP:journals/paapp/LeeTCT97,DBLP:conf/wads/Williams13}. As to formal algorithms for computing different kinds of \emph{restricted} set partitions, more recent work has been published until the year 2017, see the work of \cite{Huemer20091509,DBLP:journals/ijac/ConflittiM17,DBLP:journals/tcs/ConflittiM15}. Most lately, in 2020, the number of $\delta$-partitions of a given set has been estimated in \cite{DBLP:journals/ejc/CulverW20}. Nonetheless, we do not see in the literature any work that formalizes an algorithm for enumerating $\delta$-partitions. The so-far open issue that we close in this paper is regarding whether the $\delta$-partitions of $\alpha$ can be computed without computing the partitions of $\alpha$ in advance. As we show in this paper, this issue is not as trivial as it might seem at first sight, especially if any naive generate-and-test procedure has to be avoided. To resolve this issue, in the next section, we give a formal proof of the correctness of a smart backtracking algorithm that computes the $\delta$-partitions of $\alpha$ efficiently without a prior construction of the partitions of $\alpha$. The algorithm does so by checking (ahead) two properties identifying parts of the search space that are devoid of any $\delta$-partition of~$\alpha$.

\textbf{New results.} In section 2, we show that for any set $\alpha=\{a_1,a_2,a_3,...,a_n\}$, for any non-negative integer $\delta < n$, in the best-case scenario, the $\delta$-partitions of $\alpha$ can be computed in $\Omega(n)$ time, whereas in the worst-case scenario, the $\delta$-partitions might be computed in $\mathcal{O}(nB_n)$ time, where $B_n$ is the Bell number. In contrast, in the best and worst scenarios, a naive generate-and-test algorithm would compute the $\delta$-partitions of $\alpha$ in $\Omega(nB_n)$ time.

\section{Enumerating all $\delta$-partitions of a given set $\alpha$}
Our approach is incremental. We start with every element in a given set $\alpha=\{a_1,$ $a_2,$ $a_3,$ $...,a_n\}$ being mapped to $0$. Then, we compute $\delta$-partitions of $\alpha$ by assigning every $a_i$ to a positive integer provided that every $(a_j)_{j<i}$ is already mapped to some positive integer. Let $\Gamma_1$ and $\Gamma_2$ be partitions of $\alpha$. We say that $\Gamma_1$ and $\Gamma_2$ are \emph{identical} if and only if 
\begin{equation}
\forall j~\exists k \text{ such that }\{x : \Gamma_1(x)=j\} = \{x: \Gamma_2(x)=k\}.
\end{equation}
Obviously, if two $\delta$-partitions are identical, then it suffices to compute one of them. To this end, we use a rule for mapping elements of $\alpha$ to positive integers as specified in the next proposition.
\begin{proposition}
Let $\alpha=\{a_1,a_2,a_3,...,a_n\}$ be a set of elements, $\Gamma_1$ be a partition of $\alpha$ with $\Gamma_1(a_i)=\lambda$ for some $a_i$ such that $\lambda > \eta+1$, where
\begin{equation}
    \eta =
    \begin{cases*}
      0 & if $i=1$ \\
      m: \exists ((a_k)_{k<i},m) \in \Gamma_1 \text{ with } \forall j<i,~m\ge \Gamma_1(a_j) & if $i>1$
    \end{cases*}
  \end{equation}
Then, there exists $\Gamma_2$ identical to $\Gamma_1$ such that for every $(a_u)_{u<i}$ it is the case that $\Gamma_2(a_u)=\Gamma_1(a_u)$ and $\Gamma_2(a_i) = \eta +1$.
\end{proposition}
\begin{sproof}
Let us construct $\Gamma_2$ from $\Gamma_1$ by the following steps:
\begin{enumerate}
\item[1.]$\Gamma_2 \gets \emptyset$;
\item[2.] for each $x$ with $\Gamma_1(x)=\lambda$, $\Gamma_2 \gets \Gamma_2 \cup \{(x,\eta +1)\}$;
\item[3.] for each $x$ with $\Gamma_1(x) =\eta +1$, $\Gamma_2 \gets \Gamma_2 \cup \{(x,\lambda)\}$;
\item[4.] for each $x$ with $\Gamma_1(x) \notin \{\lambda,\eta +1\}$, $\Gamma_2 \gets \Gamma_2 \cup \{(x,\Gamma_1(x))\}$.
\end{enumerate}
Notably, this process produces $\Gamma_2$ from $\Gamma_1$ such that (1) holds, and so, $\Gamma_2$ and $\Gamma_1$ are identical. 
Due to step 2 in the process, $\Gamma_2(a_i) = \eta +1$. However, we need to show that
\begin{equation}
\forall (a_u)_{u<i},~\Gamma_2(a_u)=\Gamma_1(a_u).
\end{equation}
Observe, $\Gamma_2(a_i)=\eta +1$ and $\Gamma_1(a_i) > \eta + 1$. Hence, according to (2), for all $(a_u)_{u<i}$
\begin{equation}
\Gamma_1(a_u) < \Gamma_2(a_i) < \Gamma_1(a_i).
\end{equation}
Thus, we note that steps 2--3, in the process above, do not apply to those elements that are in $\{a_u \in \alpha : u<i\}$; but, step 4 applies to them, and so, (3) holds. \placeqed
\end{sproof}
Therefore, Proposition 1 indicates that in constructing a partition $\Gamma$ of a given set $\alpha=\{a_1,a_2,a_3,...,a_n\}$, $a_1$ is assigned only to $1$ and, every $(a_i)_{i>1}$ needs only to be assigned to a positive integer that is not greater than $m+1$, where for some $k<i$ it holds that $\Gamma(a_{k})=m$ and $m\ge \Gamma(a_j)$ for all $j<i$.

As we are concerned with computing the $\delta$-partitions of a given set, in the next two propositions we specify the mappings that can not grow to a $\delta$-partition.
\begin{proposition}
Let $\alpha=\{a_1, a_2, a_3, ..., a_n\}$ be a set of elements, $\delta <n$ be a non-negative integer, $\Gamma_1: \alpha \to \{0, 1, 2, 3, ..., n\}$ be a total mapping, and $\beta$ be a nonempty set such that
\begin{equation}
\beta = \{j>0: 0 < |\{x:\Gamma_1(x)=j\}|\le \delta\}
\end{equation}
with
\begin{equation}
|\{x:\Gamma_1(x)=0\}| < \sum_{j \in \beta} \delta - |\{x:\Gamma_1(x)=j\}|+1.
\end{equation}
Then, for all $\Gamma_2 \supseteq \{(x,j)\in \Gamma_1: j>0\}$, $\Gamma_2$ is not a $\delta$-partition of $\alpha$.
\end{proposition}
\begin{sproof}
Let $\Gamma_2=\Gamma_1$, and
\begin{equation}
\pi = \{j>0: 0 < |\{x:\Gamma_2(x)=j\}|\le \delta\},
\end{equation}
such that
\begin{equation}
|\{x:\Gamma_2(x)=0\}| < \sum_{j \in \pi} \delta - |\{x:\Gamma_2(x)=j\}|+1.
\end{equation}
Now, update $\Gamma_2$ such that for each $y$ with $\Gamma_2(y)=0$ do
\begin{enumerate}
\item[1.] $\Gamma_2(y) \gets j, \text{ for some } j \in \pi$;
\item[2.] $\pi \gets \{j>0: 0 < |\{x:\Gamma_2(x)=j\}|\le \delta\}$.
\end{enumerate}
Having applied steps 1-2 to each $y$ with $\Gamma_2(y)=0$, and given (8), we note that
\begin{equation}
0 < \sum_{j \in \pi} \delta - |\{x:\Gamma_2(x)=j\}|+1.
\end{equation}
Now, assume that $\pi=\emptyset$. Then, we rewrite (9) as $0<0$, which is impossible. Therefore, $\pi\neq \emptyset$, and so (9) implies that
\begin{equation}
\exists j>0 \text{ such that } 0<|\{x:\Gamma_2(x)=j\}| \le \delta.
\end{equation}
Hence, for all $\Gamma_2 \supseteq \{(x,j)\in \Gamma_1: j>0\}$, $\Gamma_2$ is not a $\delta$-partition of $\alpha$.\placeqed
\end{sproof}

\begin{proposition}
Let $\alpha=\{a_1, a_2, a_3, ..., a_n\}$ be a set of elements, $\delta <n$ be a non-negative integer, $\Gamma_1: \alpha \to \{0, 1, 2, 3, ..., n\}$ be a total mapping, and $\beta$ be a nonempty set such that
\begin{equation}
\beta = \{j>0: 0 < |\{x:\Gamma_1(x)=j\}|\le \delta\}
\end{equation}
with
\begin{equation}
|\{x:\Gamma_1(x)=0\}| = \sum_{j \in \beta} \delta - |\{x:\Gamma_1(x)=j\}|+1.
\end{equation}
Then, for all $\Gamma_2 \supseteq \{(x,j)\in \Gamma_1: j>0\}$ it holds that
\begin{equation}
\exists y~(\Gamma_1(y)=0 \wedge \Gamma_2(y)=k \wedge k \notin \beta) \implies \Gamma_2 \text{ is not a } \delta\text{-partition of } \alpha.
\end{equation}
\end{proposition}
\begin{sproof}
Let $\Gamma_3 = \Gamma_1$. Rewrite (11) and (12) by replacing $\Gamma_1$ with $\Gamma_3$. Thus, for
\begin{equation}
\beta = \{j>0: 0 < |\{x:\Gamma_3(x)=j\}|\le \delta\},
\end{equation}
it holds that
\begin{equation}
|\{x:\Gamma_3(x)=0\}| = \sum_{j \in \beta} \delta - |\{x:\Gamma_3(x)=j\}|+1.
\end{equation}
Now, for some $y$ with $\Gamma_3(y)=0$, set $\Gamma_3(y) \gets k$ such that $k \notin \beta$. So, rewrite (15) as
\begin{equation}
|\{x:\Gamma_3(x)=0\}| < \sum_{j \in \beta} \delta - |\{x:\Gamma_3(x)=j\}|+1.
\end{equation}
Using Proposition 2, for all $\Gamma_2 \supseteq \{(x,j)\in \Gamma_3:j>0\}$, $\Gamma_2$ is not a $\delta$-partition of~$\alpha$. Recall, $\{(x,j)\in \Gamma_3:j>0\} \supseteq  \{(x,j)\in \Gamma_1: j>0\}$. \placeqed
\end{sproof}

We now give our procedure listed in Algorithm 1. Let $\alpha=$ $\{a_1,$ $a_2,$ $a_3,$ $...,a_n\}$ be a set of elements, $\delta < n$ be a non-negative integer, and $\Gamma: \alpha \to \{0, 1, 2, ...,n\}$ be a total mapping such that for all $x \in \alpha$, $\Gamma(x)=0$. If Algorithm~1 is called with \emph{partition}($\Gamma, 1, a_1, \delta$), then the algorithm computes all $\delta$-partitions of $\alpha$.

\begin{algorithm}
	\caption{partition($\Gamma$, $\rho$, $a_i$, $\delta$)}
	$\Gamma(a_i) \gets \rho$\;	
	$\beta \gets \{j>0: 0 < |\{x:\Gamma(x)=j\}|\le \delta\}$\;	
	\If{$\beta \neq \emptyset$}{		
		$\mu \gets \sum_{j \in \beta} \delta - |\{x:\Gamma(x)=j\}|+1$\;
		\If{$\mu > |\{x:\Gamma(x)=0\}|$}{
			return\;
		}
		\If{$\mu = |\{x:\Gamma(x)=0\}|$}{
			\ForEach{$k \in \beta$}{
				partition($\Gamma, k, a_{i+1}, \delta$)\;			
			}	
			return\;
		}
	}	
	\If{$i=n$}{report that $\Gamma$ is a $\delta$-partition\; return\;}	
	Let $m \in \{x>0 \mid \exists (a_k)_{k\le i}~\Gamma(a_k)=x$ and $\forall j\le i$,~$x\ge \Gamma(a_j)\}$\;
	\ForEach{$k \gets 1$ \emph{to} $m+1$}{
		partition($\Gamma, k, a_{i+1}, \delta$)\;			
	}	
	return\;
\end{algorithm} 

\paragraph{\textbf{Example}} Run Algorithm 1 to compute all $1$-partitions of $\alpha=\{a_1,$ $a_2,$ $a_3,$ $a_4\}$. Initially, call 
\begin{equation}
\text {partition}(\{(a_1,0),(a_2,0),(a_3,0),(a_4,0)\}, 1, a_1, 1).
\end{equation}
Referring to line 15 in Algorithm 1, for $k=1$, invoke
\begin{equation}
\text {partition}(\{(a_1,1),(a_2,0),(a_3,0),(a_4,0)\}, 1, a_2, 1).
\end{equation}
Applying line 15 in Algorithm 1, when $k=1$, run
\begin{equation}
\text {partition}(\{(a_1,1),(a_2,1),(a_3,0),(a_4,0)\}, 1, a_3, 1).
\end{equation}
Following line 15 in Algorithm 1, for $k=1$, apply
\begin{equation}
\text {partition}(\{(a_1,1),(a_2,1),(a_3,1),(a_4,0)\}, 1, a_4, 1).
\end{equation}
Applying lines 11--13 in Algorithm 1, report that  $\{(a_1,1),$ $(a_2,1),$ $(a_3,1),$ $(a_4,1)\}$ is a $1$-partition of $\alpha$, and then backtrack to (19). Referring to line 15 in Algorithm~1, for $k=2$, operate
\begin{equation}
\text {partition}(\{(a_1,1),(a_2,1),(a_3,1),(a_4,0)\}, 2, a_4, 1).
\end{equation}
At this stage, as line 6 of Algorithm 1 is performed, backtrack to (19), then to (18). Executing line 15 in Algorithm 1, for $k=2$, invoke
\begin{equation}
\text {partition}(\{(a_1,1),(a_2,1),(a_3,0),(a_4,0)\}, 2, a_3, 1).
\end{equation}
Referring to line 8 in Algorithm 1, for $k=2$, run 
\begin{equation}
\text {partition}(\{(a_1,1),(a_2,1),(a_3,2),(a_4,0)\}, 2, a_4, 1).
\end{equation}
Following lines 11--13 in Algorithm 1, report that $\{(a_1,1),$ $(a_2,1),$ $(a_3,2),$ $(a_4,2)\}$ is a $1$-partition of $\alpha$, and then return to (22), (18), then back to (17).
Applying line 15 in Algorithm 1, for $k=2$, operate
\begin{equation}
\text {partition}(\{(a_1,1),(a_2,0),(a_3,0),(a_4,0)\}, 2, a_2, 1).
\end{equation}
Referring to line 8 in Algorithm 1, when $k=1$, execute
\begin{equation}
\text {partition}(\{(a_1,1),(a_2,2),(a_3,0),(a_4,0)\}, 1, a_3, 1).
\end{equation}
Performing line 8 in Algorithm 1, for $k=2$, operate
\begin{equation}
\text {partition}(\{(a_1,1),(a_2,2),(a_3,1),(a_4,0)\}, 2, a_4, 1).
\end{equation}
Following lines 11--13, report that $\{(a_1,1),(a_2,2),(a_3,1),(a_4,2)\}$ is a $1$-partition of $\alpha$, and return to (25), then to (24). Referring to line 8 in Algorithm 1, for $k=2$, run
\begin{equation}
\text {partition}(\{(a_1,1),(a_2,2),(a_3,0),(a_4,0)\}, 2, a_3, 1).
\end{equation}
According to line 8 in Algorithm 1, for $k=1$, call
\begin{equation}
\text {partition}(\{(a_1,1),(a_2,2),(a_3,2),(a_4,0)\}, 1, a_4, 1).
\end{equation}
Running lines 11--13 in Algorithm 1, report that $\{(a_1,1),$ $(a_2,2),$ $(a_3,2),$ $(a_4,1)\}$ is a $1$-partition of $\alpha$, and next, return to (27), (24), then eventually back to (17). That completes the application of Algorithm 1. 

In the next two propositions we prove the correctness of Algorithm~1.
\begin{proposition}
Let $\alpha=\{a_1,a_2,a_3,...,a_n\}$ be a set of elements, $\delta < n$ be a non-negative integer, $\Gamma: \alpha \to \{0, 1, 2, ...,n\}$ be a total mapping such that for all $x \in \alpha$ it holds that $\Gamma(x)=0$, and let Algorithm~1 be invoked with \emph{partition}$(\Gamma, 1, a_1, \delta)$. Then, for every $\Gamma$ reported by Algorithm 1, $\Gamma$ is a $\delta$-partition of $\alpha$.
\end{proposition}
\begin{sproof}
Let $\Gamma^{(i)}$ denote the mappings of $\Gamma$ at some state $i$ of Algorithm 1. From now on, whenever we say ``state'' we mean a state of Algorithm~1. The algorithm enters a new state every time line 1 is applied. Observe that the algorithm might report a $\delta$-partition at state $n+1$, see line 11 in Algorithm~1. Thus, we focus on the states that lead to a $\delta$-partition. For the initial state, it is the case that
\begin{equation}
\forall x,~\Gamma^{(0)}(x)=0.
\end{equation}
And, for every $i \in \{0,1, 2, 3, ..., n\}$ it holds that
\begin{equation}
\Gamma^{(i+1)}= (\Gamma^{(i)} \setminus \{(a_{i+1},0)\}) \cup \{(a_{i+1},\rho)\},
\end{equation}
for some positive integer $\rho$, see line 1 in Algorithm 1. Now, we need to prove that for every state, $\Gamma$ can grow to a $\delta$-partition as described in Proposition 2 and 3. Starting with the case described in Proposition 2, we need to prove that for every state $i \in \{0, 1, 2, 3, ..., n\}$, it holds that
\begin{equation}
|\{x:\Gamma^{(i)}(x)=0\}| > \sum_{j \in \beta^{(i)}} \delta - |\{x:\Gamma^{(i)}(x)=j\}|+1,
\end{equation}
where 
\begin{equation}
\beta^{(i)}= \{j>0: 0 < |\{x:\Gamma^{(i)}(x)=j\}|\le \delta\}.
\end{equation}
For the initial state, $\beta^{(0)}=\emptyset$, recall (29). Thereby, (31) holds since 
\begin{equation}
\sum_{j \in \beta^{(0)}} \delta - |\{x:\Gamma^{(0)}(x)=j\}|+1 = 0
\end{equation}
and, recalling (29),
\begin{equation}
|\{x:\Gamma^{(0)}(x)=0\}|=n.
\end{equation}
Now, we will show that for all $i$, if (31) is true for $i$, then (31) holds for $i+1$. Assume that (31) holds for $i$. Using (32), we note that
\begin{equation}
\beta^{(i+1)}= \{j>0: 0 < |\{x:\Gamma^{(i+1)}(x)=j\}|\le \delta\}.
\end{equation}
Recall (30); if $\rho \in \beta^{(i+1)}$, then the inequality in (31) holds for $i+1$ since $(a_{i+1},0)$ is dropped from $\{x:\Gamma^{(i)}(x)=0\}$ in the left side of the inequality and $(a_{i+1},\rho)$ is added to $\{x:\Gamma^{(i)}(x)=\rho\}$ in the right side. Thus,
\begin{equation}
\{x:\Gamma^{(i+1)}(x)=0\}=\{x:\Gamma^{(i)}(x)=0\} \setminus \{(a_{i+1},0)\},
\end{equation}
and 
\begin{equation}
\{x:\Gamma^{(i+1)}(x)=\rho\}=\{x:\Gamma^{(i)}(x)=\rho\} \cup \{(a_{i+1},\rho)\}.
\end{equation}
Therefore,
\begin{equation}
|\{x:\Gamma^{(i+1)}(x)=0\}| > \sum_{j \in \beta^{(i+1)}} \delta - |\{x:\Gamma^{(i+1)}(x)=j\}|+1.
\end{equation}
Referring to (30), if $\rho \notin \beta^{(i+1)}$, then the inequality (31) might not hold for $i+1$; hence, it might be the case that
\begin{equation}
|\{x:\Gamma^{(i+1)}(x)=0\}| < \sum_{j \in \beta^{(i+1)}} \delta - |\{x:\Gamma^{(i+1)}(x)=j\}|+1,
\end{equation}
or
\begin{equation}
|\{x:\Gamma^{(i+1)}(x)=0\}| = \sum_{j \in \beta^{(i+1)}} \delta - |\{x:\Gamma^{(i+1)}(x)=j\}|+1.
\end{equation}
If (39) is the case, which is illustrated in Proposition 2, then according to lines 5--6, Algorithm~1 returns to a previous state to pick another positive integer to be assigned to $a_{i+1}$.
If (40) is the case, which is demonstrated in Proposition 3, the algorithm assigns $a_{i+2}$ to a positive integer from $\beta^{(i+1)}$, see lines 7--9 in Algorithm~1. \placeqed
\end{sproof}

\begin{proposition}
Let $\alpha=\{a_1,a_2,a_3,...,a_n\}$ be a set of elements, $\delta < n$ be a non-negative integer, $\Gamma: \alpha \to \{0, 1, 2, ...,n\}$ be a total mapping such that for all $x \in \alpha$ it holds that $\Gamma(x)=0$, and let Algorithm~1 be invoked with \emph{partition}$(\Gamma, 1, a_1, \delta)$. Every $\delta$-partition of $\alpha$ is computed by the algorithm.
\end{proposition}
\begin{sproof}
Assume that $\varphi$ is a $\delta$-partition of $\alpha$, but $\varphi$ is not computed by the algorithm. Thus, for every state $s$
\begin{equation}
\varphi \text{ is not identical to } \Gamma^{(s)}.
\end{equation}
We consider the case where $\varphi$ is not identical to any $\delta$-partition computed by the algorithm; otherwise, there is no need to compute such $\varphi$ as demonstrated in Proposition 1. Our purpose is to establish a contradiction with (41). Thus, we define
\begin{equation}
\Lambda^{(r)}_i = \{1, 2, 3, ..., m+1\}
\end{equation}
to be the set of positive integers for some $a_i \in \alpha$ at some state $r$ such that $m=0$ if $i=1$; for $i>1$, $m=\Gamma^{(r)}(a_k)$ for some $k<i$ such that $m\ge \Gamma^{(r)}(a_j)$ for all $j<i$. Additionally, we define $\sigma^{(t)}_i \subseteq \Lambda^{(t)}_i$ to be the set of positive integers for some $a_i \in \alpha$ at some state $t$ such that for every state $s$ it is the case that
\begin{equation}
\forall a_i ~\Gamma^{(s)}(a_i) \in \sigma^{(s)}_i\iff \Gamma^{(s)} \text{ is a } \delta\text{-partition of } \alpha.
\end{equation}
Since $\varphi$ is a $\delta$-partition of $\alpha$, it is the case that for all positive integers~$i$
\begin{equation}
1\le \varphi(a_i)\le m+1,
\end{equation}
where $m=0$ if $i=1$; for $i>1$, $m=\varphi(a_{k})$ for some $k<i$ such that $m\ge \varphi(a_j)$ for all $j<i$.
Considering (42) and (44), we note that
\begin{equation}
\exists s~ \forall a_i~\varphi(a_i) \in \Lambda^{(s)}_i,
\end{equation}
which is in line with Proposition 1 that is implemented at lines 15--16 in the algorithm. Referring to (43),
\begin{equation}
\exists s \forall a_i ~\varphi(a_i) \in \sigma^{(s)}_i;
\end{equation}
otherwise $\varphi$ is not a $\delta$-partition as established in Proposition 2 \& 3 and implemented at lines 2--10 in the algorithm. Following (43),
\begin{equation}
\exists s ~\forall a_i ~\varphi(a_i) = \Gamma^{(s)}(a_i) \text{ and } \Gamma^{(s)} \text{ is a } \delta\text{-partition of } \alpha,
\end{equation}
which is consistently with Proposition 4 that shows the correctness of Algorithm~1. Subsequently,
\begin{equation}
\exists s ~\varphi = \Gamma^{(s)} \text{ and } \Gamma^{(s)} \text{ is a } \delta\text{-partition of } \alpha.
\end{equation}
See the contradiction between (48) and (41). \placeqed
\end{sproof}
\begin{proposition}
Let $\alpha=\{a_1,a_2,a_3,...,a_n\}$ be a set of elements, $\delta < n$ be a non-negative integer, $\Gamma: \alpha \to \{0, 1, 2, ...,n\}$ be a total mapping such that $\Gamma(x)=0$ for all $x \in \alpha$, and let Algorithm~1 be invoked with \emph{partition}$(\Gamma, 1, a_1, \delta)$. Then,\\
(1) Algorithm 1 runs in $\Theta(n)$ space. \\
(2) Algorithm 1 runs in $\Omega(n)$ time in the best case scenario. \\
(3) Algorithm 1 runs in $\mathcal{O}(nB_n)$ time in the worst-case scenario, where $B_n$ is the Bell number.
\end{proposition}

\begin{proof}
(1) Apart from the linear space (stack) required to execute the recursion of Algorithm 1, and the linear space imposed by using the mapping $\Gamma$ inputted to the algorithm, an additional space is required. Note, the set $\beta$ (line 2 in the algorithm) can be implemented instead as a total mapping $\mathcal{B}: \{1,2,3, ..., n\} \to \{0,1,2,3,...,n\}$ such that in the initial state of the algorithm we set $\mathcal{B}(j)=0$ for all $j$. Thus, we can update $\mathcal{B}(j) \gets \mathcal{B}(j) +1$ in constant time whenever we map an element of $\alpha$ to some part $j$, see line 1 in the algorithm. Consequently, for any state of the algorithm, for any $j$ it holds that \begin{equation}\mathcal{B}(j)=|\{x:\Gamma(x)=j\}|.\end{equation} Now we show that $\mu$ (line 4 in the algorithm) can be computed in constant time. Hence, in the initial state of Algorithm 1 we set $\mathcal{M} =0$. Whenever we apply $\mathcal{B}(k) \gets \mathcal{B}(k) +1$ for some $k$, $\mathcal{M}$  is updated subsequently (in constant time) as follows
\begin{equation}
    \mathcal{M} =
    \begin{cases*}
      \mathcal{M} - 1 & if $\mathcal{B}(k) >1$ \\
      \mathcal{M} + \delta & if $\mathcal{B}(k)=1$
    \end{cases*}
\end{equation}
We now prove inductively that $\mathcal{M}$ is equivalent to $\mu$ that is employed by Algorithm 1 at line 4. For the initial state of Algorithm 1, it is obvious that $\mu^{(0)} = \mathcal{M}^{(0)}$. Next, we demonstrate that for all states $i$ it is the case that 
\begin{equation}
\mathcal{M}^{(i)}=\mu^{(i)} \implies \mathcal{M}^{(i+1)}=\mu^{(i+1)}.
\end{equation}
Suppose $\mathcal{M}^{(i)}=\mu^{(i)}$. According to line 4 in the algorithm, we write
\begin{equation}
\mathcal{M}^{(i)} = \mu^{(i)} =  \sum_{j \in \{l>0:0<|\{x:\Gamma^{(i)}(x)=l\}|\le \delta\}} \delta - |\{x:\Gamma^{(i)}(x)=j\}|+1.
\end{equation}
If we set $\Gamma^{(i+1)}(x)=k$ for some $x \in \alpha$ at some state $i+1$ where $\mathcal{B}^{(i)}(k) > 0$, then $\mathcal{B}^{(i+1)}(k)>1$. Therefore, using (50) we have
\begin{equation}
\mathcal{M}^{(i+1)}=\mathcal{M}^{(i)} - 1 =\mu^{(i)} -1.
\end{equation}
Observe, 
\begin{equation}
|\{x:\Gamma^{(i+1)}(x)=k\}|=|\{x:\Gamma^{(i)}(x)=k\}| + 1.
\end{equation}
Thus,
\begin{equation}
\mathcal{M}^{(i+1)} =\mu^{(i)} - 1= -1 + \sum_{j \in \{l>0:0<|\{x:\Gamma^{(i)}(x)=l\}|\le \delta\}} \delta - |\{x:\Gamma^{(i)}(x)=j\}|+ 1.
\end{equation}
Using (54), (55) can be rewritten as
\begin{equation}
\mathcal{M}^{(i+1)} =  \sum_{j \in \{l>0:0<|\{x:\Gamma^{(i+1)}(x)=l\}|\le \delta\}} \delta - |\{x:\Gamma^{(i+1)}(x)=j\}|+1 =\mu^{(i+1)}.
\end{equation}
For the second case of (50), if we set $\Gamma^{(i+1)}(x)=k$ for some $x \in \alpha$ at some state $i+1$ with $\mathcal{B}^{(i)}(k)=0$, then $\mathcal{B}^{(i+1)}(k)=1$. Therefore, using (50) we write
\begin{equation}
\mathcal{M}^{(i+1)}=\mathcal{M}^{(i)} + \delta=\mu^{(i)} + \delta.
\end{equation}
Observe, 
\begin{equation}
\mu^{(i+1)} = \sum_{j \in \{l>0:0<|\{x:\Gamma^{(i+1)}(x)=l\}|\le \delta\}} \delta - |\{x:\Gamma^{(i+1)}(x)=j\}|+1.
\end{equation}
Subsequently,
\begin{equation}
\mu^{(i+1)}= (\delta - |\{x:\Gamma^{(i+1)}(x)=k\}|+1) + \mu^{(i)} = \delta + \mu^{(i)}=\mathcal{M}^{(i)} + \delta=\mathcal{M}^{(i+1)}.
\end{equation}
Note that \begin{equation}|\{x:\Gamma^{(i+1)}(x)=k\}|=\mathcal{B}^{(i+1)}(k)=1.\end{equation}
\\(2) Apply Algorithm 1 with $\delta = n-1$. Then, the algorithm will perform operations in the order of $n$ until the only $\delta$-partition of $\alpha$ is produced. Thus, the lower-bound time $\Omega(n)$ follows immediately. \\
(3) If $\delta=0$, then the algorithm has to compute all partitions of $\alpha$. It is well known that the number of partitions of $n$ elements is equal to the Bell number $B_n$. Since constructing each partition might take at maximum $n$ operations (i.e. assigning $n$ elements to parts, see line 1 in Algorithm 1), the overall running time of Algorithm 1 is upper-bounded by $\mathcal{O}(nB_n)$.
\end{proof}

\section{Concluding Remarks}
\paragraph{Remark 1} Let $\alpha=\{a_1,a_2,a_3,...,a_n\}$ be a set of elements, and $\delta < n$ be a non-negative integer. It follows directly from Proposition 6 that Algorithm~1 generally is faster than any algorithm that generates all partitions of $\alpha$ for computing the $\delta$-partitions of $\alpha$. Observe that for $\delta = 0$, the $\delta$-partitions of $\alpha$ coincide with the partitions of $\alpha$. Likewise, it is worth noting that if we drop the lines 2--10 from Algorithm 1, then we end up with an algorithm that computes all partitions of $\alpha$. Nonetheless, the running-time efficiency achieved by Algorithm~1 comes at the expense of an additional $\Theta(n)$ space.
\paragraph{Remark 2}  A related open problem is to decide an asymptotic lower-bound time for computing different kinds of restricted set partitions, such as set partitions with certain groups of elements being mapped to the same integers (i.e. parts), or the contrary, being mapped to strictly different integers.

\bibliographystyle{elsarticle-num}
\bibliography{mybibfile}

\end{document}